\documentclass[]{article}
\usepackage{amssymb}

\usepackage{amsfonts}
\usepackage{graphicx}
\usepackage{amsthm}
\newtheorem{theorem}{Theorem}

\newtheorem{corollary}[theorem]{Corollary}

\newtheorem{definition}[theorem]{Definition}
\newtheorem{example}[theorem]{Example}

\newtheorem{lemma}[theorem]{Lemma}

\newtheorem{proposition}[theorem]{Proposition}

\newtheorem{remarks}[theorem]{Remarks}

\begin{document}
\title{
Synchronization in discrete-time networks with general pairwise
coupling \\{\small Published in Nonlinearity \textbf{22} (2009)
2333-2351.}}



\author{}
\author{}
\author{Frank Bauer\thanks{\texttt{bauer@mis.mpg.de}}
\and Fatihcan M. Atay\thanks{\texttt{fatay@mis.mpg.de};
http://private-pages.mis.mpg.de/fatay} \and J\"urgen
Jost\thanks{\texttt{jjost@mis.mpg.de}} }

\date{Max Planck Institute for Mathematics in the Sciences\\
 Inselstrasse 22, D-04103 Leipzig, Germany}

\maketitle

\begin{abstract}
We consider complete synchronization of identical maps coupled
through a general interaction function and in a general network
topology where the edges may be directed and may carry both positive
and negative weights.
We define mixed transverse exponents and derive sufficient
conditions for local complete synchronization. The general
non-diffusive
coupling scheme can lead to new synchronous behavior, in networks of
identical units, that cannot be produced by single units in
isolation. In particular, we show that synchronous chaos can emerge
in networks of simple units. Conversely, in networks of chaotic
units simple synchronous dynamics can emerge; that is, chaos can be
suppressed through synchrony.
\end{abstract}

\section{Introduction}

Synchronization in complex networks has been studied extensively by
many scientists in the past years (for recent reviews see
\cite{Arenas08,Chavez07,Pikovsky01}). Specific application areas
include social networks (opinion formation, finance, and world trade
web), biological networks (genetic networks, cardiac rhythms, and
neural networks), and technological networks (wireless communication
networks and power-grids); see \cite{Arenas08} and the references
therein. Synchronization is one particular collective behavior in
complex networks that emerges through the interaction of the
constituent units. Complex systems are generally characterized by
the richness of emergent  behavior arising from the interaction of
many  elements or agents, which themselves are typically rather
simple and often interact only locally or with only a few other
ones. How dynamically rich behavior can emerge in a network of
simple units is an important general question in complexity. In this
article, we study the relation of the concept of emergence to
synchronization\footnote{There exits several notions of
synchronization \cite{Pikovsky01}. In this work we study complete
synchronization, where the differences between the states of coupled
units tend to zero.}, in the setting of coupled identical map
networks.

Coupled map networks introduced by Kaneko \cite{Kaneko84} have
become one of the standard models in synchronization studies.
Particularly, chaotic synchronization can be investigated already in
simple one-dimensional maps, which, as is well-known, would require
at least three dimensions in the continuous-time case. On the one
hand, a chaotic system's sensitive dependence to initial conditions
would tend to lead the coupled system away from synchrony in the
presence of ever-so-small perturbations. On the other hand, a
diffusive coupling \footnote{Diffusive coupling refers to a coupling
function that vanishes whenever its arguments are the same, see
(\ref{13}).} between the units tends to equalize the states of
neighboring units by providing a driving force that grows with the
difference of the states, thereby driving the system towards
synchrony. In the interplay between these two effects, it has indeed
been discovered that diffusively coupled chaotic systems can exhibit
robust synchronization for an appropriate range of parameters.
Synchronization in diffusively-coupled map networks with positive
weights is now well understood: One finds that the synchronizability
of the network depends on the underlying network topology (given by
the eigenvalues of the coupling matrix) and the dynamical behavior
of the individual units (given by the largest Lyapunov exponent)
\cite{Jost01,Lu04}.

A notable aspect of diffusively coupled systems of identical units
is that, while diffusive coupling tends to drive the network
towards synchrony, the synchronized network shows exactly the same
dynamical behavior as a single isolated unit. This is a simple
consequence of the fact that the coupling term vanishes at the
synchronized state. Hence, normally no new behavior arises through
synchronization of diffusively-coupled systems. A notable
exception is in the presence of time delays: It has been
discovered that networks with time delays not only are able to
synchronize (in fact sometimes better than undelayed networks),
but also can exhibit a very rich range of new synchronized
behavior \cite{Atay06, Atay04b}.

However, new synchronized behavior can not only be observed in
diffusively coupled networks, when time-delays are present. In
this work we show that also in non-diffusively coupled networks
new synchronized behavior is emerging. The focus of this work is
to understand, in detail, the impact of non-diffusive coupling
schemes on the synchronized behavior. The synchronized behavior of
a non-diffusively coupled network can be very different from the
individual dynamics of its components (units). In contrast to
diffusively-coupled networks, the overall coupling strength does
not only determine the robustness of the synchronous state but
also the emerging synchronous dynamics, i.e. by changing the
overall coupling strength, different synchronous dynamics can be
exhibited by the same system. In particular we shall see the
following two extreme cases. Synchronized chaos can emerge in
networks of simple units and chaos can be suppressed in a network
of chaotic units.




\section{General pairwise coupling \label{19}}
We consider a network of identical units that interact with each
other. The coupling topology will affect the resulting dynamics in
such a network. In many real world applications the connection
structure is not bidirectional. Furthermore the influence of other
units can be excitatory or inhibitory. These phenomena can, for
example, be observed in neural networks, where there exists
excitatory and inhibitory synapses and only the pre-synaptic neuron
influences the post-synaptic one but not vice versa. Examples of
synchronization in simple models of neural networks are investigated
in \cite{Bauer08b}. Since different networks interact through
different interaction functions, we want to keep our coupling
function sufficiently general. Taking all these requirements into
account, we study the following network model.

Let $\Gamma$ be a non-trivial, weighted, directed graph on $n$
vertices. The weight of the connection from vertex $j$ to vertex $i$
is denoted by $w _{ij}$, which could be positive, negative or zero.
Positive and negative weights could model, for instance, excitatory
and inhibitory connections, respectively. We assume that the network
has no self-loops, that is, $w _{ii}=0$ for all $i$. The in-degree
of vertex $i$ is denoted by $d_i = \sum_{j=1}^n w _{ij}$. Even if a
vertex is not isolated, it is possible that the in-degree of this
vertex is equal to zero because of cancellations between positive
and negative weights. These vertices are called quasi-isolated
vertices, because their properties are very similar to isolated
vertices \cite{Bauer08}. Note that by definition every isolated
vertex is quasi-isolated. In the following we will identify each
unit with a vertex of a graph. This correspondence allows us to make
use of graph theoretical methods.

The activity at vertex or unit $i$ at time $t+1$ is given by:
\begin{equation}
x_i(t+1) = f(x_i(t)) + \epsilon \sigma_i\sum_{j=1}^n w
_{ij}g(x_i(t),x_j(t)) \qquad  i = 1,...,n\label{1},
\end{equation}
where
$$
\sigma_i = \left\{
\begin{array}{l c l}\frac{1}{d_i} & \mbox{if} & d_i \neq 0 \\
 0 & \mbox{if} & d_i = 0,
\end{array} \right.
$$
$f: \mathbb{R} \rightarrow \mathbb{R}$ and $g: \mathbb{R}^2
\rightarrow \mathbb{R}$ are 
differentiable functions, and $\epsilon \in \mathbb{R}$ is the
overall coupling strength. In the following we assume that the
derivatives of $f$ and $g$ are bounded along synchronous solutions
$s(t)$ of system (\ref{1}), where all units exhibit the same
behavior for all time i.e., \[x_i(t) = s(t) \;\;\forall \,i,\,t.\]
The function $f$ describes the dynamical behavior of the
individual units whereas $g$ characterizes the interactions
between different pairs of units.

We say that the interaction is diffusive if the coupling function
$g$ satisfies the general diffusion condition
\begin{equation}\label{13} g(x,x) = 0 \;\; \forall \,x\in \mathbb{R}.
\end{equation}
A synchronous solution always exists if the interaction function $g$
satisfies (\ref{13}). Under this assumption the dynamical behavior
of the whole network is exactly the same as the dynamical behavior
of the isolated units, that is,
\begin{equation} \label{46} s(t+1) = f(s(t)).
\end{equation}
However, if (\ref{13}) is not satisfied then a synchronous solution
may not always exist. Indeed, if $s(t)$ is a synchronous solution
such that $g(s(t),s(t))\neq 0$ for some $t$, then (\ref{1}) implies
that either $\sigma_i=0$ for all $i$, or $\sigma_i\neq 0$ for all
$i$. In other words, a synchronized solution exists only if all
vertices are quasi-isolated or none of them are. The first case is
trivial as there is no interaction.
Therefore, for the study of emergent dynamics when (\ref{13}) is not
satisfied, we shall later on restrict ourselves to the second case,
i.e., to networks without quasi-isolated vertices. For such networks
synchronous solutions exist and satisfy
\begin{equation} \label{25} s(t+1) = f(s(t)) + \epsilon g(s(t),s(t)).
\end{equation}
Eq.~(\ref{25}) already shows that the synchronous behavior of the
network is different from the
behavior of isolated units. 
However, at this point it is not clear whether a synchronized state
is robust against perturbations.  We will study this issue in some
detail in the following sections.

It is important to note that the interactions between the
different units in Eq.~(\ref{1}) are ``normalized" by the factor
$\sigma_i$. Otherwise a synchronous solution would only exist, in
the non-diffusive case, under the assumption that all vertices
have the same vertex in-degree. Consequently a synchronous
solution would only exist for regular graphs.

\section{The coupling matrix \label{20}}
For diffusively coupled networks, the graph Laplacian is the natural
coupling matrix. For a directed, weighted graph without loops, the
(normalized) graph Laplacian $\mathcal{L}$  is defined as,
\[(\mathbf{\mathcal{L}})_{ij} := \left\{
\begin{array}{r cl} 1 & \mbox{ if} & i = j \mbox{ and $d_i \neq 0$}. \\
 -\frac{w _{ij}}{d_i} &\mbox{ if}& \mbox{there is a directed edge from }j \mbox{ to }i\mbox{  and } d_i\neq 0. \\
0&&
 \mbox{otherwise.}
\end{array} \right.\]
Here in the more general coupling case, the natural coupling matrix
is given by $\mathbf{K}$, where $\mathbf{K}$ for a directed,
weighted network without loops, is defined as
\[(\mathbf{K})_{ij} := \left\{
\begin{array}{r cl} 1 & \mbox{ if}& i = j \mbox{ and $d_i = 0$}. \\
 \frac{w _{ij}}{d_i} & \mbox{ if}& \mbox{there is a directed edge from
   }j \mbox{ to }i\mbox{ and } d_i\neq 0. \\
0&&
 \mbox{otherwise.}
\end{array} \right.\]
Let the eigenvalues of $\mathbf{K}$ and $\mathcal{L}$ be labeled
as $\lambda_1,...,\lambda_n$ and $ \lambda^\prime_1,...,
\lambda^\prime_n$, respectively. Since the row sums of
$\mathbf{K}$ are equal to $1$, $\mathbf{K}$ has always an
eigenvalue equal to 1, which  corresponds to the eigenvector
$\mathbf{e} = (1,...,1)^\top$.

We briefly discuss the relationship between the coupling matrix
$\mathbf{K}$ and the graph Laplacian $\mathcal{L}$ for directed
weighted graphs. The graph Laplacian $\mathbf{\mathcal{L}}$ and the
coupling matrix $\mathbf{K}$ are related to each other by
\[\mathbf{\mathcal{L}} = \mathbf{I} - \mathbf{K},\] where
$\mathbf{I}$ is the $n\times n$ identity matrix.  Thus we have
\begin{equation}
\label{11} \lambda^\prime_i = 1 - \lambda_i, \;\;\forall \, i.
\end{equation}
Hence the multiplicity of the zero eigenvalue of $\mathcal{L}$ is
equal to the multiplicity of the eigenvalue $\lambda=1$ of
$\mathbf{K}$.

The spectral properties of $\mathcal{L}$ and $\mathbf{K}$ for
directed graphs with mixed signs are investigated systematically in
\cite{Bauer08}. The presence of directed edges or mixed signs leads
to some interesting differences in the spectrum of $\mathcal{L}$
compared to the case of undirected edges and nonnegative weights.
For convenience of the reader, we mention here some of these
differences. For undirected graphs with nonnegative weights it is
well-known that all eigenvalues of $\mathcal{L}$ are real
\cite{Chung97}. However, this is not true anymore if one studies
directed graphs or mixed
signs. 
Furthermore, while in the case of only nonnegative weights the
absolute values of the eigenvalues are bounded by 2 \cite{Chung97},
this is no longer true for the case of mixed signs. Indeed, using
Gershgorin's Theorem \cite{Horn06} we have the following estimate.

\begin{lemma}\label{48}  Let $\mathcal{D}(c,r)$ denote
the disk in the complex plane centered at $c$ and having radius $r$.
Assume that there is at least one non quasi-isolated vertex
(otherwise $\mathcal{L}$ is equal to the zero matrix). Then all
eigenvalues of the graph Laplacian $\mathcal{L}$ are contained in
the disk $\mathcal{D}(1,r)$, where
\begin{equation}\label{r} r := \max_i \frac{\sum_{j=1}^n\left|w
_{ij}\right|}{\left|\sum_{j=1}^nw _{ij}\right|} = \max_i
\frac{\sum_{j=1}^n\left|w
_{ij}\right|}{\left|d_i\right|},\end{equation} with the convention
$\left|d_i\right|^{-1} = 0$ if $d_i = 0$.
\end{lemma}

Note that the radius $r$ in Eq.~(\ref{r}) can be written in the form
\[r = \max_i \left|\frac{d_i^+ + d_i^-}{d_i^+ - d_i^-}\right|,\]
where $d_i ^+ := \sum_{j:w_{ij} \geq 0}w_{ij}$ is the positive
in-degree and $d_i ^- := \sum_{j:w_{ij} \leq 0}|w_{ij}|$ the
negative in-degree. So, for fixed sum $d_i^+ + d_i^-$, the closer
$d_i^+$ and $d_i^-$ are, the larger is the radius $r$. Clearly,
$r=1$ when the weights are nonnegative, but $r$ can be much larger
in the case of signed weights. The radius $r$ will play an
important role in Chapter \ref{57} where we derive sufficient
conditions for synchronization.

Since we deal with a graph $\Gamma$ with both positive and
negative weights, we should make some notions precise. Let the
weighted adjacency matrix of $\Gamma$ be given by $\mathbf{W}=
[w_{ij}]$. Define the corresponding (usual) graph $\Gamma^\prime$
with adjacency matrix $\mathbf{A} = [a_{ij}]$ such that $a_{ij}
=1$ if $w_{ij} \neq 0$ and $a_{ij} =0$ otherwise. Then we say that
$\Gamma$ is strongly connected (resp. has a spanning tree) if
$\Gamma^\prime$ is strongly connected (resp. has a spanning tree).
In the general setting of directed edges and mixed signs, a
non-complete graph refers to a graph where there exists at least
one pair of distinct vertices with no link between them.

The multiplicity of the zero eigenvalue of $\mathcal{L}$ can be
bounded from below by the following two graph properties:

\begin{lemma}\label{17}
Let $m_0$ be the multiplicity of the zero eigenvalue
 of $\mathcal{L}$.
\begin{enumerate}
\item If $n_1$ denotes the number of quasi-isolated vertices, then
\[n_1 \leq m_0.\]
\item \label{24} If  $n_2$ denotes the minimum number of trees
needed to span the graph, then
\[n_2 \leq m_0.\]
\end{enumerate}
\end{lemma}
\begin{proof}
1. Assume that there exists $n_1$ quasi-isolated vertices in the
graph. Then, there exists $n_1$ rows of $\mathcal{L}$ that consist
entirely of zeros. Consequently, $\mathcal{L}$ has at least $n_1$
eigenvalues equal to zero. \\\\2. Assume that one needs at least
$n_2$ trees to span the whole graph. Let $\mathcal{L}$ be given in
Frobenius normal form \cite{Brualdi91}, i.e. after possibly
relabeling the vertices, $\mathcal{L}$ is given in the form
\begin{equation} \mathcal{L} = \left(
\begin{array}{cccccc} \mathcal{L}_1 & \mathcal{L}_{12} &...  &\mathcal{L}_{1p} &
\\0& \mathcal{L}_2&...&\mathcal{L}_{2p} \\\vdots & \vdots& \ddots & \vdots \\  0& 0& ...&
\mathcal{L}_p
\end{array} \right),
\end{equation} where the block diagonal matrices $\mathcal{L}_i$
correspond to the strongly connected components of the graph. The
spectrum of $\mathcal{L}$ then satisfies
\begin{equation} \label{eq:spec}  \mbox{spec}(\mathcal{L}) =
\bigcup_{i=1}^p \mbox{spec}(\mathcal{L}_i).
\end{equation} Since one needs $n_2$ trees to span the whole graph,
there exist $n_2$ block matrices $\mathcal{L}_i$ such that
$\mathcal{L}_{ij}$ is the zero matrix for $i<j\leq p$. Thus, these
$n_2$ block matrices $\mathcal{L}_i$ have zero row sums because
$\mathcal{L}$ has zero row sums. The result now follows from
(\ref{eq:spec}).
\end{proof}

\section{Synchronization}
We want to study synchronous solutions of Eq.~(\ref{1}) and whether
the synchronous state is robust to perturbations. We say the system
(\ref{1}) (locally) synchronizes if
\[\lim_{t \rightarrow \infty} |x_i(t) - x_j(t)| = 0 \quad \forall \,
i,j,\] whenever the initial conditions  belong to some appropriate
open set\footnote{If one considers chaotic synchronization, i.e. $f
+ \epsilon g$ is chaotic, then there exists  subtleties concerning
this open set and the exact notion of attraction. These issues are
carefully studied in \cite{Lu07}.  For the purposes of this article,
these subtleties are not important.}. In this article the term
synchronization always refers to this definition.

\subsection{General pairwise coupling and only quasi-isolated vertices}
This case is not very insightful, as Eq.~(\ref{1}) implies that
there are no interactions between the different units. The time
evolution of each unit is given by Eq.~(\ref{46}). Thus the
network only synchronizes if nearby orbits of the function $f$
converge, i.e. the Lyapunov exponent \begin{equation} \label{Lyap}
\mu_f :=
\overline{\lim}_{T\rightarrow\infty}\frac{1}{T}\sum_{s=\bar{t}}^{\bar{t}+T-1}\log|f'(s(t))|,\end{equation}
of $f$ is negative. Here $\bar{t}$ is chosen such that
$f'(s(t))\neq 0$ for all $t>\bar{t}$. A negative Lyapunov exponent
implies that the trajectory $s(t)$ is already attracting for $f$,
and hence, no emergence of new dynamics. Hence chaotic
synchronization is not possible as chaos requires a positive
Lyapunov exponent $\mu_f$.

\subsection{General pairwise coupling without quasi-isolated
vertices}To characterize this case we start with the following
definition.
\begin{definition}The $k$-th mixed
transverse exponent $\chi_k$ is defined for $2\leq k \leq n$ as:
\begin{equation} \chi_k
:=\overline{\lim}_{T\rightarrow\infty}
\frac{1}{T}\sum_{s=\bar{t}}^{\bar{t}+T-1}\log|h_k(s(t))|,\label{30}
\end{equation} where \[h_k(s(t))=f'(s(t)) + \epsilon
\partial_1g(s(t),s(t)) + \epsilon
\partial_2g(s(t),s(t))\lambda_k\]and $\bar{t}$ is chosen such that $h(s(t))\neq
0$ for all $t>\bar{t}$. If no such $\bar{t}$ exists we set $\chi_k =
- \infty$.
\end{definition}
These exponents combine the dynamical behavior of the functions $f$
and $g$ with the network topology. Furthermore we define the
\textit{maximal mixed transverse exponent $\chi$} as
\[\chi := \max_{k \geq 2} \chi_k. \]
\\
The next theorem shows that the maximal mixed transverse exponent
governs the synchronizability of the network.

\begin{theorem}\label{15}
System (\ref{1}) synchronizes if the maximal mixed transverse
exponent is negative.
\end{theorem}
Before we prove this theorem we prove the following theorem that
also holds for time-dependent functions.
\begin{theorem}
\label{thm:linear-convergence}Consider the system of equations
\begin{equation}v_{i}(t+1)=\left\{\begin{array}{lc}
k_1(t)v_{i}(t), & i=1, \\
 k_1(t)v_{i}(t)+k_2(t)v_{i-1}(t), & i=2,\dots,m,
\end{array}\right.\label{eq:a1}\end{equation}
where $v_{i}\in\mathbb{R}$ and $k_1$ and $k_2$ are bounded functions
on $\mathbb{R}$. Suppose there exists $\bar{t}\in\mathbb{R}$ such
that
\begin{equation} k_1(t)\neq0\mbox{ for
}t\ge\bar{t}.\label{eq:tbar}\end{equation}
 Suppose further that \begin{equation}
\eta:=\overline{\lim}_{T\rightarrow\infty}\frac{1}{T}\sup_{t_{0}\geq\bar{t}}\sum_{s=t_{0}}^{t_{0}+T-1}\log|k_1(s)|<0.\label{eq:mu}\end{equation}
Then for any $\varepsilon\in(0,-\eta)$ there exists
$K_{\varepsilon}\geq0$ such that all solutions of (\ref{eq:a1})
satisfy\begin{equation} \Vert(v_{1}(t),\dots,v_{m}(t))\Vert\leq
K_{\varepsilon}e^{(\eta+\varepsilon)(t-t_{0})}\Vert(v_{1}(t_{0}),\dots,v_{m}(t_{0}))\Vert\label{eq:result}\end{equation}
for all $t\ge t_{0}\geq\bar{t}.$ On the other hand, if there is no
such $\bar{t}$ satisfying (\ref{eq:tbar}), then \begin{equation}
\Vert(v_{1}(t),\dots,v_{m}(t))\Vert=0\quad\mbox{for all large
}t.\label{eq:result2}\end{equation}
\end{theorem}
\begin{proof}
The homogeneous equation\begin{equation}
v_1(t+1)=k_1(t)v_1(t)\label{eq:hom}\end{equation} has the solution
$v_1(t)=\Phi(t,t_{0})v_1(t_{0})$, where the state transition
function $\Phi$ is given by
$\Phi(t,t_{0})=\prod_{s=t_{0}}^{t-1}k_1(s)$, $t>t_{0},$ and
$\Phi(t,t)=1$ for $\forall t$. In the following, $t_{0}\ge\bar{t}$.
By (\ref{eq:mu}), for any $\varepsilon\in(0,-\eta)$ there exists
$T^{\prime}$ such that \[
\sup_{t_{0}\ge\bar{t}}\frac{1}{T}\sum_{s=t_{0}}^{t_{0}+T-1}\log|k_1(s)|<\eta+\frac{\varepsilon}{m}<0,\quad\mbox{for
all }T> T^{\prime}.\] Thus, \[ \prod_{s=t_{0}}^{t-1}|k_1(s)|\leq
e^{(\eta+\frac{1}{m}\varepsilon)(t-t_{0})},\quad\mbox{if }t>
t_{0}+T^{\prime},\]
 whereas \[
\prod_{s=t_{0}}^{t-1}|k_1(s)|\leq M^{(t-t_{0})}\leq
M^{T^{\prime}},\quad\mbox{if }t_{0}<t\leq t_{0}+T^{\prime}\]
 where $M=\sup_{t\in\mathbb{R}}|k_1(t)|$. Thus, \[
\prod_{s=t_{0}}^{t-1}|k_1(s)|\leq
C_{1}^{\varepsilon}e^{(\eta+\frac{1}{m}\varepsilon)(t-t_{0})},\quad\mbox{\ensuremath{\forall}}t>t_{0}\ge\bar{t}\]
where the constant
$C_{1}^{\varepsilon}\ge\max\{M^{T^{\prime}}e^{-(\eta+\frac{1}{m}\varepsilon)T^\prime},1\}$
is independent of $t_{0}$. Consequently, \begin{equation}
|\Phi(t,t_{0})|\leq
C_{1}^{\varepsilon}e^{(t-t_{0})(\eta+\frac{1}{m}\varepsilon)},\quad\forall
t>t_{0}\ge\bar{t},\label{eq:est}\end{equation} and the solution of
the first equation in (\ref{eq:a1}) satisfies\begin{equation}
|v_{1}(t)|\le
C_{1}^{\varepsilon}e^{(\eta+\frac{1}{m}\varepsilon)(t-t_{0})}|v_{1}(t_{0})|.\label{eq:hom-est}\end{equation}

Now the solution to (\ref{eq:a1}) for $i=2,\dots,m$
is\begin{equation}
v_{i}(t)=\Phi(t,t_{0})v_{i}(t_{0})+\sum_{s=t_{0}}^{t-1}\Phi(t,s+1)k_2(s)v_{i-1}(s).\label{eq:soln}\end{equation}
Using (\ref{eq:est}) and (\ref{eq:hom-est}) we
estimate,\begin{eqnarray*} |v_{2}(t)| & \leq
C_{1}^{\varepsilon}e^{(\eta+\frac{1}{m}\varepsilon)(t-t_{0})}|v_{2}(t_{0})|
+C_{1}^{\varepsilon}\overline{k_2}\sum_{s=t_{0}}^{t-1}e^{(\eta+\frac{1}{m}\varepsilon)(t-s-1)}
e^{(\eta+\frac{1}{m}\varepsilon)(s-t_{0})}|v_{1}(t_{0})|\\
 & =C_{1}^{\varepsilon}e^{(\eta+\frac{1}{m}\varepsilon)(t-t_{0})}|v_{2}(t_{0})|+
 C_{1}^{\varepsilon}\overline{k_2}e^{(\eta+\frac{1}{m}\varepsilon)(t-t_{0})}
 e^{-(\eta+\frac{1}{m}\varepsilon)}(t-t_{0})|v_{1}(t_{0})|,\end{eqnarray*}
 where $\overline{k_2}=\sup_t |k_2(t)|$. Adding to (\ref{eq:hom-est}) yields\begin{eqnarray*}
|v_{2}(t)|+|v_{1}(t)| & \le
C_{1}^{\varepsilon}e^{(\eta+\frac{1}{m}\varepsilon)(t-t_{0})}|v_{2}(t_{0})|+C_{1}^{\varepsilon}
e^{(\eta+\frac{1}{m}\varepsilon)(t-t_{0})}(\overline{k_2}e^{-(\eta+\frac{1}{m}\varepsilon)}(t-t_{0})+1)|v_{1}(t_{0})|.\end{eqnarray*}
Since there exists some constant $C_{2}$ (we drop the dependence on
$\varepsilon$ for ease of notation) such that \[
C_{1}^{\varepsilon}e^{(\eta+\frac{1}{m}\varepsilon)(t-t_{0})}(\overline{k_2}e^{-(\eta+\frac{1}{m}\varepsilon)}(t-t_{0})+1)\le
C_{2}e^{(\eta+\frac{2}{m}\varepsilon)(t-t_{0})},\quad\forall t>
t_{0},\] we have\begin{equation} |v_{1}(t)|+|v_{2}(t)|\le
C_{2}e^{(\eta+\frac{2}{m}\varepsilon)(t-t_{0})}(|v_{1}(t_{0})|+|v_{2}(t_{0})|),\quad\forall
t> t_{0}\ge\bar{t}.\label{eq:v2}\end{equation}

For $i=3,$ the argument is similar with some slight modifications:
We use the estimate from (\ref{eq:est}) as\[ |\Phi(t,t_{0})|\leq
C_{1}e^{(t-t_{0})(\eta+\frac{1}{m}\varepsilon)}\le
C_{1}e^{(t-t_{0})(\eta+\frac{2}{m}\varepsilon)}\]
 in (\ref{eq:soln}), while bounding $|v_{2}(t)|$ by the right hand
side of (\ref{eq:v2}). Thus,\[ |v_{3}(t)|\leq
C_{1}e^{(\eta+\frac{2}{m}\varepsilon)(t-t_{0})}|v_{3}(t_{0})|+C_{2}
\sum_{s=t_{0}}^{t-1}e^{(\eta+\frac{2}{m}\varepsilon)(t-s-1)}\overline{k_2}
e^{(\eta+\frac{2}{m}\varepsilon)(s-t_{0})}(|v_{1}(t_{0})|+|v_{2}(t_{0})|).\]
Adding to (\ref{eq:v2}) gives\[ \sum_{i=1}^{3}|v_{i}(t)|\le
C_{3}e^{(\eta+\frac{3}{m}\varepsilon)(t-t_{0})}\sum_{i=1}^{3}|v_{i}(t_{0})|\]
for some constant $C_{3}$. Repeating for $i=4,\dots,m$, we finally
obtain\[
\]
\[
\sum_{i=1}^{m}|v_{i}(t)|\le
C_{m}e^{(\eta+\frac{m}{m}\varepsilon)(t-t_{0})}\sum_{i=1}^{m}|v_{i}(t_{0})|\]
which establishes (\ref{eq:result}) for the $\ell_{1}$-norm, and
thus for all norms in $\mathbb{R}^{m}$ for an appropriate constant
$K_{\varepsilon}$.

To prove the last statement of the theorem, notice that if
(\ref{eq:tbar}) fails for all $t$, then there exist an infinite
sequence $t_{1}<t_{2}<\cdots$ of zeros of $k_1$. In this case, the
equations (\ref{eq:a1}) imply that $v_{i}(t)=0$ for all $t\ge
t_{i}+1$, yielding (\ref{eq:result2}). This completes the proof.
\end{proof}
Now we prove Theorem \ref{15}.
\begin{proof}[Proof of Theorem \ref{15}]
\label{Sec 4.2.} For general pairwise coupling and networks without
quasi-isolated vertices, Eq.~(\ref{1}) can be written in the
following form using the coupling matrix $\mathbf{K}$
\begin{equation}\label{31} x_i(t+1) = f(x_i(t)) + \epsilon
\left(\mathbf{K}\mathbf{g}(\mathbf{x_i}(t),\mathbf{x}(t))\right)_i,
\end{equation}
where the vector $\mathbf{g(x_i,x)}\in \mathbb{R}^n$ is defined as
$\mathbf{g(x_i,x)}:= (g(x_i,x_1),..., g(x_i,x_n))^\top$ and
$\left(\mathbf{K}\mathbf{g}(\mathbf{x_i}(t),\mathbf{x}(t))\right)_i$
is the $i$th component of the vector
$\mathbf{K}\mathbf{g}(\mathbf{x_i}(t),\mathbf{x}(t))$.

Let $\mathbf{x}(t) = (x_1(t), \ldots, x_n(t))$ and $
\mathbf{s}(t)=(s(t),\ldots, s(t))$. Small perturbations
$\mathbf{u}(t) = \mathbf{x}(t) - \mathbf{s}(t)$ of the synchronous
state are governed by the variational equation
\begin{equation}\label{29} \mathbf{u}(t+1) = \left[f^\prime(s(t)) +
\epsilon \partial_1g(s(t),s(t))\right]\mathbf{u}(t) + \epsilon
\partial_2g(s(t),s(t))\mathbf{Ku}(t),
\end{equation}
where $\partial_ig$ denotes the $i$th partial derivative of $g$.

In the sequel we study the coupling matrix $\mathbf{K}$ in Jordan
form. There exists a non-singular matrix $\mathbf{P}$ such that
$\mathbf{K} = \mathbf{PJP}^{-1}$ and $\mathbf{J}$ is of the form
\begin{equation}
 \mathbf{J}  = \left( \begin{array}{l}\mathbf{J}_1
\\ \quad \;\; \mathbf{J}_2 \qquad \quad  \\   \qquad   \qquad \ddots \\    \qquad \quad  \qquad \quad   \mathbf{J}_m
\end{array} \right).
\end{equation}
Each Jordan block $\mathbf{J}_l$ is of the form
\begin{equation} \indent \indent \indent \indent \mathbf{J}_l =
 \left( \begin{array}{*{2}{l@{\quad }l}} \lambda_l  & 1  & & \\ & \ddots&
 \ddots \\ & & &1\\ & & &\lambda_l
\end{array} \right) \in \mathbb{R}^{m_l \times m_l},
\end{equation}
where $m_l$ is the block size of the Jordan block $J_l$. Without
loss of generality we assume that $\mathbf{J}_1$ corresponds to
the eigenvalue $\lambda_1=1$ with eigenvector $\mathbf{e} :=
(1,...,1)^\top$. After the coordinate transformation,
$\mathbf{u}(t) \rightarrow \mathbf{P}^{-1}\mathbf{u}(t) =:
\mathbf{v}(t)$, Eq.~(\ref{29}) becomes:
\begin{equation} \mathbf{v}(t+1) = \left[f^\prime(s(t)) +
\epsilon \partial_1g(s(t),s(t))\right]\mathbf{v}(t) + \epsilon
\partial_2g(s(t),s(t))\mathbf{Jv}(t).
\end{equation}
For each Jordan-block $\mathbf{J}_l$ this  reads in component form:
\begin{equation} \label{12}
v_i(t+1) = \left\{
\begin{array}{l@{\quad \quad}r}
h_l(s(t))v_i(t) + \epsilon
\partial_2g(s(t),s(t))v_{i+1}(t) & i = 1,...,m_{l}-1\\
h_l(s(t))v_i(t) & i = m_l\end{array} \right.
\end{equation}
This is exactly of the form (\ref{eq:a1}) with $k_1 =h_l$ and $k_2
= \epsilon\partial_2g$. We now apply theorem
\ref{thm:linear-convergence}, noting that the time dependance in
(\ref{12}) arises from a trajectory of a time-invariant system, so
choices of initial times are arbitrary. Thus, the perturbations
decay for all Jordan blocks $\mathbf{J}_i$, $i= 2,...,m$. For the
Jordan block $\mathbf{J}_1$ the situation is different because we
do not assume that the mixed longitudinal exponent $\chi_1$
(similarly defined as the mixed transverse exponents $\chi_k$ in
(\ref{30}) for the eigenvalue $\lambda_1 = 1$) is negative. Thus,
$v_1(t)$ does not have to decay. However, we are mainly interested
in the behavior of the original perturbations $\mathbf{u}(t) =
\mathbf{Pv}(t)$. Since $\lambda_1 = 1$ corresponds to the
eigenvector $\mathbf{e} := (1,...,1)^\top$ it is possible to
choose $\mathbf{P}$ such that
\begin{equation}  \mathbf{P} =
 \left( \begin{array}{*{2}{l@{\quad }l}} 1  &p_{12}  &\ldots & p_{1n}\\1&p_{22}&\ldots&p_{2n}\\\vdots &\vdots &
&\vdots
\\1 & p_{n2}&\ldots &p_{nn}
\end{array} \right).
\end{equation}
Thus, $\mathbf{u}(t) = \mathbf{Pv}(t)$ is given by
\[\mathbf{u}(t) = v_1(t)  \mathbf{e} + \left( \begin{array}{c}
\sum_{j=2}^n p_{1j}v_j(t)
\\\vdots\\\sum_{j=2}^n p_{nj}v_j(t)
\end{array} \right).\]  Since all the $v_i(t)$, $i = 2,...,n$ are
decaying we conclude that
\begin{eqnarray*}\lim_{t\rightarrow \infty}|x_i(t)-x_j(t)| &=&
 \lim_{t\rightarrow \infty}|s(t) + u_i(t) - (s(t)
+u_j(t)) | \\&=& \lim_{t\rightarrow \infty}\left|s(t) + v_1(t) +
\sum_{k=2}^n p_{ik}v_k(t) - \left(s(t) + v_1(t) + \sum_{k=2}^n
p_{jk}v_k(t)\right)\right|\\&=&\lim_{t\rightarrow \infty} \left|
\sum_{k=2}^n p_{ik}v_k(t) -  \sum_{k=2}^n p_{jk}v_k(t)\right|= 0
\;\forall \, i,j\;.\end{eqnarray*}
\end{proof}
\begin{remarks}
\begin{itemize}\item []
\item If in addition $\chi_1<0$ then $v_1(t)$ is also decaying. Thus
the synchronous solution is attracting. 
Examples of attracting synchronous solutions are studied in Chapter
\ref{57}.
\item Theorem \ref{15} can also be formulated in terms of the eigenvalues
of the graph Laplacian $\mathcal{L}$. In this case one only has to
replace $\lambda_k$ by $1 - \lambda^\prime_k$, according to equation
(\ref{11}).
\end{itemize}
\end{remarks}

In the following we study some special cases of coupling functions
that appear commonly in applications.

\subsection{Diffusive coupling with both quasi-isolated and non quasi-isolated vertices}
For diffusive coupling a synchronous solution always exists. So,
here we may permit the coexistence of both
 non-quasi-isolated and quasi-isolated vertices in the
network. The general diffusion condition (\ref{13}) implies that
$\partial_1 g(s(t),s(t)) = -
\partial_2g(s(t),s(t))$. Because this leads to a cancellation of the
artificially introduced term in the coupling matrix $\mathbf{K}$ (1
when $i=j$ and $d_i=0$), it follows that, even when there are
non-quasi-isolated and quasi-isolated vertices in the graph, small
perturbations $\mathbf{u}(t) = \mathbf{x}(t) - \mathbf{s}(t)$ are
governed again by the variational equation (\ref{29}), and the same
arguments apply as in section \ref{Sec 4.2.}. In this case the
\textit{k-th mixed transverse exponent for diffusively coupled
units} is given by
\begin{eqnarray} \chi^{\mbox{diff}}_k:= \overline{\lim}_{T\rightarrow\infty}
\frac{1}{T}\sum_{s=\bar{t}}^{\bar{t}+T-1}\log|h^{\mbox{diff}}_k(s(t))|\end{eqnarray}where
\[h^{\mbox{diff}}_k(s(t)) =f'(s(t))) + \epsilon
\partial_2g(s(t),s(t))(\lambda_k-1)\] and $\bar{t}$ is chosen such that $h(s(t))\neq
0$ for all $t>\bar{t}$. If no such $\bar{t}$ exists we set
$\chi_k^{diff} = - \infty$.

For diffusive coupling, Theorem \ref{15} implies that system
(\ref{1}) synchronizes if the maximal transverse exponent satisfies
\begin{equation} \label{43}
\chi^{\mbox{diff}}:= \max_{k \geq 2}\chi^{\mbox{diff}}_k < 0.
\end{equation}

\begin{proposition} \label{45} Assume that the function $f$ is chaotic,
i.e. has a positive Lyapunov-exponent $\mu_f$. If
\begin{itemize}
\item there exists more than one quasi-isolated vertex \\\\ or
\item the network does not possess a spanning tree
\end{itemize}
then the maximal mixed transverse exponent is positive.
\end{proposition}

\begin{proof}
If one of these conditions if fulfilled then by Lemma \ref{17} the
multiplicity of the zero eigenvalue satisfies $m_0(\mathcal{L}) \geq
2$ and hence the multiplicity of the eigenvalue $1$ of $\mathbf{K}$
satisfies $m_1(\mathbf{K})\geq2$. Consequently,
$\chi_k^{\mbox{diff}} = \mu_f
> 0$ for $k = 2,...,m_0$.
\end{proof}
\medskip

In fact, it is intuitively clear that the presence of more than one
quasi-isolated vertex or the absence of a spanning tree will in
general make chaotic synchronization impossible, because in those
situations, there will exist pairs of vertices none of which can
dynamically influence the other. Note, however, that Proposition
\ref{45} does not exclude  the so-called Master-Slave
configurations.

A particular coupling function that arises in coupled map lattice
models \cite{Kaneko-book93} is
\begin{equation}\label{39}  g(x_i,x_j) = b(f(x_j) -
f(x_i))\end{equation} where $b$ is some real constant. In this case
it is possible to separate the effects of the synchronous dynamics
$f$ and the network topology:

\begin{corollary}\label{diff-coupling}
System (\ref{1}) with the coupling function (\ref{39}) synchronizes
if
\begin{equation} \label{47}
\mu_f + \max_{k \geq 2} \;\log|1+ \epsilon b (\lambda_k-1)|<0. 
\end{equation}
\end{corollary}
This result was already obtained in undirected \cite{Jost01} and
 directed \cite{Lu04} networks, in both cases with nonnegative weights.

Assume that the synchronous solution $f$ is chaotic, i.e. $\mu_f>0$.
Then the network topology term in (\ref{47}) has to be sufficiently
negative to compensate the positive Lyapunov exponent $\mu_f$. This
in turn requires that the eigenvalues $\lambda_k$ for $k\ge 2$ be
bounded away from one, and the coupling strength $\epsilon$ lie in
an appropriate interval.

\subsection{Direct coupling and only non-quasi-isolated vertices}

Another important special case of the general coupling function
$g(x,x)$ is  the so-called direct coupling\footnote{We borrow this
term from \cite{Aronson90}.}, where the interactions depend only on
the state of the neighboring units, i.e.
\begin{equation}\label{40}  g(x_i,x_j) = \hat{g}(x_j)
\end{equation}
for some $\hat{g}:\mathbb{R} \rightarrow \mathbb{R}$. Thus
$\partial_1g(s(t),s(t)) = 0$, and
the \textit{k-th mixed transverse exponent for directly coupled
units} is given by
\begin{eqnarray} \chi^{\mbox{direct}}_k:= \overline{\lim}_{T\rightarrow\infty}
\frac{1}{T}\sum_{s=\bar{t}}^{\bar{t}+T-1}\log|h^{\mbox{direct}}_k(s(t))|\end{eqnarray}where
\[h^{\mbox{direct}}_k(s(t)) =f'(s(t))) + \epsilon
\hat{g}'(s(t))\lambda_k.\] and $\bar{t}$ is chosen such that
$h(s(t))\neq 0$ for all $t>\bar{t}$. If no such $\bar{t}$ exists we
set $\chi_k^{diff} = - \infty$.
For direct coupling, Theorem \ref{15} implies that system (\ref{1})
synchronizes if the maximal mixed transverse exponent satisfies
\begin{equation}  \label{44} \chi^{\mbox{direct}}:=\max_{k \geq2} \chi_k^{direct} < 0.
\end{equation}

%
\begin{proposition}\label{42} Assume that
$f + \epsilon \hat{g}$ is chaotic, i.e has a positive
Lyapunov-exponent $\mu_{(f+\epsilon g)}$. If the network does not
possess
a spanning tree then 
maximal mixed transverse exponent $ \chi^{\mbox{direct}}$ is
positive.
\end{proposition}
\begin{proof} Similar to the proof of Proposition \ref{45}.
\end{proof}

If in addition $\hat{g}(x) = bf(x)$, it is again possible to
separate the effects of the resulting synchronized dynamics $(1 +
\epsilon b)f$ and the network topology.
\begin{corollary}\label{36} Suppose that $\hat{g}(x) = bf(x)$.
System (\ref{1}) synchronizes if
\begin{equation} \label{8}
\max_{k \geq 2} \;\log\left|\frac{1+ \epsilon b
\lambda_k}{1+\epsilon b}\right| + \mu_{(1+ \epsilon b)f}  < 0.
\end{equation}
\end{corollary}

Similar to the case of diffusively coupled units, if the resulting
synchronous solution $(1 + \epsilon b)f$ is chaotic, then  the
eigenvalues $\lambda_k$ should be bounded away from one, and the
coupling strength $\epsilon$ lie in an appropriate interval.

\section{Emergence of synchronized chaos in directly coupled networks \label{22}}\label{6}
As mentioned in the introduction, in diffusively-coupled units, the
synchronized network behaves exactly as a single unit would do in
isolation. 
Hence, no new collective behavior is gained from synchronization:
Either the units are dynamically complex, then so is the whole
network, - or  the units are dynamically simple then the behavior of
the whole network will remain simple. 

Eq.~(\ref{25}) shows that the requirement for new emerging
collective behavior is that the general diffusion condition
Eq.~(\ref{13}) is not satisfied. 
 In this way,  the network can display new behavior that the single
isolated units are not capable to show. In particular we shall see
in this section that synchronous chaos can emerge in a network of
simple units. 

One task in showing that simple units can exhibit synchronized chaos
is to rigorously show that the synchronized solution is indeed
chaotic. To this end, we consider S-unimodal maps.

\subsection{A full family of S-unimodal maps}

In general it is hard to prove that a function is chaotic. However
there is one well-understood class of chaotic functions, namely the
so-called full families of S-unimodal maps \cite{Devaney89}. For
convenience, we recall the notion of a full family of S-unimodal
maps in the Appendix \ref{Appendix}.

\textbf{Properties of full families of S-unimodal maps:}\\It is
well-known that a full family of S-unimodal maps undergoes a
sequence of period doubling bifurcations as $\mu$ varies from
$\mu_0$ to $\mu_1$ and finally becomes chaotic \cite{Devaney89,
Whitley83}. \\\\As the main result of this section, we present a
full family of S-unimodal maps.

\begin{theorem}\label{38} Consider the family $f_\mu$  of maps
\begin{equation}\label{54}
f_\mu(x) = \left( \frac{1 -e^{-\mu}}{1 + e^{-\mu}} x+ 1
-\frac{2}{1+e^{-\mu x}}\right)\Theta,\label{41}
\end{equation}
where $\Theta < 0$. This family is a full family of S-unimodal maps
on [0,1] with $\mu \in [\mu_0 = 0,\mu_1]$ where we choose $\mu_1$
such that $f_{\mu_1}(c) = 1$.
\end{theorem}The proof is given in Appendix \ref{Appendix B}.
\\\\
We can use the full family of S-unimodal maps (\ref{54}) to
demonstrate that synchronous chaotic behavior can emerge in a
network of simple units. Let the individual dynamics be given by
$f(x) = \left(\frac{1-e^{-\mu}}{1+ e^{-\mu}}x + 1\right) \Theta$ and
the interaction between units by  $\hat{g} (x)=
-\frac{2\Theta}{1+e^{-\mu x}}$. Clearly, the individual dynamics are
very simple, as $f(x)$ has one single attracting fixed point. The
resulting synchronous solution is \begin{equation}\label{58} s(t+1)
= \left(\frac{1-e^{-\mu}}{1+ e^{-\mu}}s(t) + 1\right)\Theta -
\frac{2\epsilon\Theta}{1+e^{-\mu s(t)}}.
\end{equation}
If $\epsilon =1$ and $\Theta <0$ this is exactly of the form
(\ref{54}). Fig.~\ref{Fig.14} supports our finding in Theorem
\ref{38} and shows that Eq.~(\ref{58}) is chaotic for a wide range
of $\epsilon$-values.

For $\Theta>0$ the family of maps (\ref{54}) is not a full family of
S-unimodal maps. However, computer simulations indicate that even in
this case the Lyapunov exponent of (\ref{54}) can be positive
depending on the parameter values of $\Theta, \mu$ and $\epsilon$.
For positive values of $\Theta$, Eq.~(\ref{54}) can be used to model
neural networks \cite{Bauer08b}.

\begin{figure}\label{Fig.14}
\begin{center}
{\includegraphics[width = 8cm]{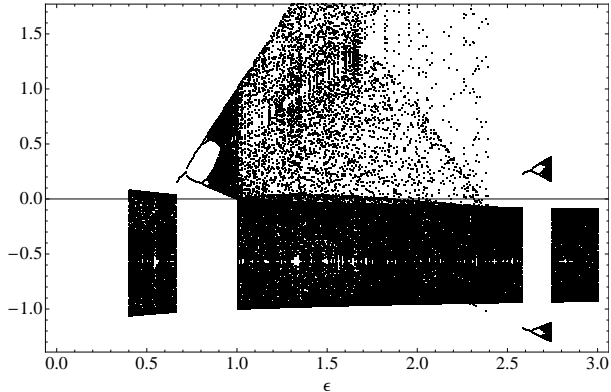}}
\end{center}
\caption{ Bifurcation diagram for Eq.~(\ref{41}) with parameter
values $\Theta = -1.3041$ and $\mu = 20$.}
\end{figure}

Another example of a full family of S-unimodal maps is given by the
familiar logistic maps. In the next section we study in detail the
emergent synchronous chaotic dynamics in networks of coupled
logistic maps.

\subsection{Coupled logistic maps}

It is well-known that the logistic map \[\ell_\rho(x) = \rho x(1-x),
\;\; \rho\in[0,4] \mbox{ and } x\in [0,1]\] is a full family
of S-unimodal maps. 
The logistic map is probably the best analyzed chaotic map; still
not everything is understood rigorously. For convenience we briefly
recall some properties. 
For $\rho = 2$ the dynamics of the logistic map are very simple,
with the fixed point $x = 1/2$ attracting all points in the open
interval $(0,1)$. The logistic map $\ell_\rho$ undergoes a period
doubling route to chaos \cite{Devaney89}. At $\rho = 3$ the first
period-doubling occurs, followed by further period-doubling
bifurcations with increasing values of $\rho$, which accumulate at
$\rho \approx 3.57$. For $\rho > 3.57$ the logistic map can be
chaotic but there are also so-called periodic windows in the
parameter interval $\rho \in (3.57, 4]$. For $\rho = 4$ it is
maximally chaotic with a Lyapunov exponent $\mu_{\ell_4} = \ln 2$.


We consider a network of coupled logistic maps. Let $f(x) =
a_1x(1-x)$, $\hat{g}(x) = a_2x(1-x)$ and
\begin{equation} \label{26} \epsilon =
\frac{\rho-a_1}{a_2}.\end{equation} Assume that $a_1, a_2\; \in
(0,4)$ and $x(0) \in [0,1]$. Then the dynamics of the synchronized
solution is given by
\begin{equation} \label{27} s(t+1) = \ell_\rho(s(t)) = \rho
s(t)(1-s(t)).\end{equation} For given $a_1$ and $a_2$  the coupling
constant $\epsilon$ can be used to control the dynamics of the
synchronous solution, i.e. $\epsilon$ can be used as bifurcation parameter. 
In particular, synchronous chaotic behavior can emerge in the whole
network, even if the individual dynamics are very simple.

By an application of Corollary \ref{36}, the directly coupled
network of logistic maps synchronizes if
\begin{equation}\label{16} \max_{k \geq 2} \log\left|1 -
\left(1-\frac{a_1}{\rho}\right)\lambda^\prime_k\right|+
\mu_{\ell_\rho} < 0.\end{equation} For illustration we consider some
concrete examples.

\begin{example}\label{32}
Let $f(x) = \hat{g}(x) = \ell_2(x)$. Choosing $\epsilon = 1$ yields
$\rho = 4$, and the synchronous solution $s(t) = \ell_4(t)$ becomes
maximally chaotic. Thus, the whole network displays complicated
dynamics although each unit of the network itself is dynamically
simple.

Eq.~(\ref{16}) implies that the network synchronizes if all
eigenvalues of the graph Laplacian (except $\lambda_1' = 0$) are
contained in $\mathcal{D}(2,1)$.  Within the class of undirected
graphs with nonnegative weights, the latter condition can only be
satisfied for complete graphs \cite{Chung97}. However, for directed
graphs or in the case of mixed signs there exist also non-complete
graphs satisfying this condition \cite{Atay08}.
\end{example}
In the next example we choose different functions $f$ and $\hat{g}$
that lead to the same synchronous solution. Interestingly, the
condition on the eigenvalues is this time different than in the case
of Example \ref{32}.

\begin{example}\label{55}Let $f(x) = \ell_1(x)$ and $\hat{g} = \ell_3(x)$. Choosing
$\epsilon = 1$ implies that the synchronous solution is the same as
in Example \ref{32}, i.e. $s(t) = \ell_4(x)$.

The network synchronizes if all eigenvalues (except  $\lambda_1' =
0$) of the graph Laplacian are contained in
$\mathcal{D}(\frac{4}{3}, \frac{2}{3})$. In contrast to Example
\ref{32}, there also exists non-complete undirected graphs with
positive weights that satisfy the latter condition.
\end{example}

Comparing Example \ref{32} and Example \ref{55} shows one
interesting point. In both Examples the synchronous solution
$s(t)$ and the overall coupling strengths are identical. However,
in Example \ref{32} the network synchronizes if all eigenvalues
are contained in $\mathcal{D}(2,1)$ whereas the same is true for
Example \ref{55} if all eigenvalues are contained in
$\mathcal{D}(\frac{4}{3}, \frac{2}{3})$.  Thus, in contrast to
diffusively coupled units, the synchronizability of a directly
coupled network is not completely determined by the synchronous
solution $s(t)$, the overall coupling strength $\epsilon$ and the
network topology (i.e. the eigenvalues of the coupling matrix),
but depends also on the special choices of $f$ and $\hat{g}$.
Clearly, this can already be seen from the definition of
$\chi_k^{\mbox{direct}}.$

\section{Suppression of chaos}\label{57}

Besides the emergence of chaos in  a network of simple units, the
opposite is also possible. In this section we show that
non-diffusive coupling can also be used to suppress chaos. 
Chaos suppression in single systems is a well-established field; for
an overview see \cite{Schöll07} and the references therein. In our
setting, suppression arises in networks through synchronization of
the states of the units.

As a first example, we consider a network of directly coupled
chaotic logistic maps.

\begin{example} \label{51}
Let $f(x) = \hat{g}(x) = \ell_4(x)$ be given and choose $\epsilon =
\left(\frac{2.1}{4} -1\right)$. This choice implies that the
synchronous behavior is given by $s(t) = \ell_{2.1}(t)$ and is
non-chaotic. Corollary \ref{36} implies that the network
synchronizes if all eigenvalues of the graph Laplacian are contained
in $\mathcal{D}(-1.11, 10.99)$.
\end{example}

Example \ref{51} shows that the network synchronizes for a wide
range of eigenvalues $\lambda^\prime_k$. Compared to synchronous
chaotic behavior, studied in Example \ref{32} and \ref{55}, the
radius of the disk $\mathcal{D}$ is much larger for simple
synchronous behavior. The reason can be seen from condition
(\ref{8}):
A smaller  Lyapunov exponent for the synchronous behavior implies
less
restrictions on the allowable values of $\lambda_k'$. 

Lemma \ref{48} implies that all eigenvalues of $\mathcal{L}$ are
contained in the disk $\mathcal{D}(1,r)$. For sufficiently small $r$
the latter disk is contained in $\mathcal{D}(-1.11, 10.99)$. So,
instead of calculating all eigenvalues of the coupling matrix, the
knowledge of the radius $r$ can be sufficient to determine whether
the network synchronizes. Hence, the quantity $r$ already gives some
insights concerning the robustness of the synchronous state.

Motivated by the above example, we will give in the next section
sufficient conditions for synchronization in terms of $r$. In fact,
these conditions will not depend on the eigenvalues of the coupling
matrix. Again, we want to point out that this is only possible if
the synchronous behavior is not chaotic.



\subsection{Direct coupling}
We restrict ourselves to the case of direct coupling where
$\hat{g}(x) = bf(x)$. It follows from Corollary \ref{36} and Lemma
\ref{48} that the network synchronizes if
\begin{equation} \label{223} \mu_{(1+
\epsilon b)f} < \log \left| \frac{1 + \epsilon b}{1 + |\epsilon
b|r}\right|.\end{equation}  Again, we see that a simpler synchronous
behavior in the sense of a smaller Lyapunov exponent implies that
the network synchronizes for a large class of network topologies.

\begin{example}\label{56} Let $f(x) = \hat{g}(x) = \ell_4(x)$ and choose $\epsilon = \left(
\frac{\rho}{4} -1\right)$. This choice implies that $s(t) =
\ell_\rho(t)$. By Eq.~(\ref{223}), the network synchronizes if
\begin{equation} \label{224}
\mu_{\ell_\rho} +\log\left|\frac{4}{\rho} +
\left(\frac{4}{\rho}-1\right)r\right|< 0.
\end{equation} The left-hand-side of (\ref{224}) is plotted in Figure \ref{Fig.23}.
\begin{itemize}
\item For $\rho = 2$, $x^* = 1/2$ is an attracting fixed point
of $s(t) = \ell_2(t)$ that attracts all points in the open interval
$(0,1)$. For $x(t=0)=x_0 \in (0,1)$ we have
\[x(t) = 1/2[1-(1-2x_0)^{2t}].\] 
Because of this fast convergence to $\frac{1}{2}$, the fact that
$\ell^\prime_2(\frac{1}{2}) = 0$ implies that $\mu_{\ell_2} = -
\infty$.\footnote{Because of this property, some authors call the
fixed point $\frac{1}{2}$ superstable. \cite{Froyland92}.} Thus the
network synchronizes for all network topologies.

\item For $\rho = 2.1$ and $\rho = 3.25$ the solid and dot-dashed  line in Fig.~\ref{Fig.23} show
that the network synchronizes  for all network topologies that
satisfy $r <8$ or $r<12$ respectively. In comparison to Example
\ref{51} here we only use the information given by $r$ instead of
using all eigenvalues of the coupling matrix. \item For $\rho = 2.5$
the dashed line in Fig.~\ref{Fig.23} shows that the function
$\mu_{\ell_{2.5}} +\log\left|\frac{4}{2.5} +
\left(\frac{4}{2.5}-1\right)r\right|$ is never negative. Hence
estimates based on $r$ are too crude to obtain synchronization
information, however in such a case we can use Corollary \ref{36} to
conclude that the network synchronizes if all eigenvalues
$\lambda_k^\prime$ of $\mathcal{L}$ (except $\lambda_1^\prime = 0$)
are contained in a disk centered at $-5/3$ with a radius $10/3$.
\end{itemize}

\begin{figure}
\begin{center}
\includegraphics[width = 10cm]{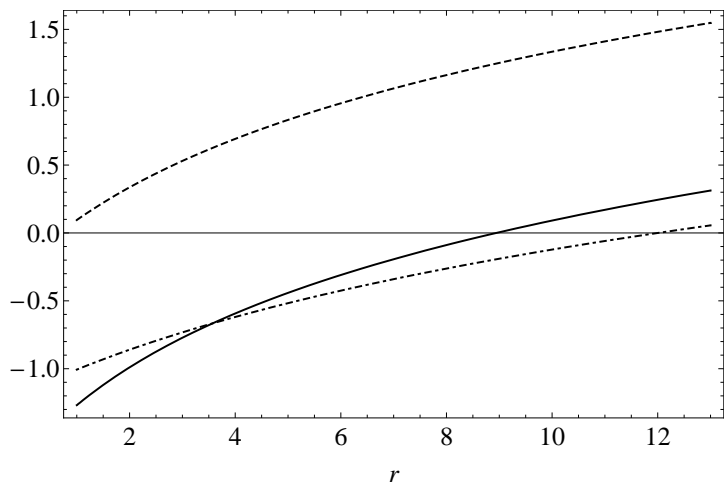}
\caption{\label{Fig.23}Plot of  $\mu_{\ell_{\rho}} +
\log\left|\frac{4}{\rho} + \left(\frac{4}{\rho} -1\right)r\right|$
as a function of $r$. The dot-dashed line corresponds to the
parameter value $\rho = 3.25$, the dashed line to $\rho = 2.5$ and
the solid line to $\rho = 2.1$. }
\end{center}
\end{figure}
\end{example}

Finally, we give an example of synchronization independent of
network topology when all the weights have the same sign.

\begin{example}
Let $\hat{g}(x) = bf(x)$ be  arbitrarily maps. Assume that  all
weights in the network are either nonnegative or nonpositive, i.e.
$r = 1$. If $\epsilon b\geq 0$ then condition (\ref{223}) is
satisfied if
\[\mu_{((1+\epsilon b)f)} < \log 1 = 0.\] If $\epsilon b< 0$ then condition (\ref{223}) is
satisfied if \[\mu_{((1+\epsilon b)f)} < \log \left| \frac{1 +
\epsilon b}{1 -\epsilon b}\right|.\]
\end{example}

\section{Conclusion and discussion \label{23}}

In this work we studied complete synchronization in coupled map
networks of identical units. We have generalized synchronization
analysis to directed networks with both positive and negative
weights and general pairwise
coupling functions. 
This generalization is especially important for analyzing real-world
networks. For example, in neural networks the coupling function is
 non-diffusive. Excitation and inhibition are modeled by positive and negative
weights, respectively, and only the pre-synaptic neuron influences
the post-synaptic one but not vice versa. Thus, directed networks
with signed weights are needed to model neural networks in an
appropriate way.

We have derived sufficient conditions for local synchronization,
which are expressed in terms of the maximal mixed transverse
exponent. We have shown that, in contrast to diffusively coupled
units, non-diffusively coupled networks can display new synchronized
behavior that the individual unit is not able to show. The new
synchronous behavior is influenced by the overall coupling strength
that plays the role of a bifurcation parameter. Again, this is
important in neural networks. Changing the coupling strength in a
learning process allows the network to learn new, possibly much
richer behavior than before.  In particular,
 synchronized chaotic behavior can arise
in networks of non-chaotic units with non-chaotic interactions
functions. Conversely, chaos can be suppressed through
synchronization in networks of chaotic units. This may have
practical implications in the field of chaos control. Often chaos is
controlled by applying time-delayed external feedback. In our
approach chaos can be controlled by changing an internal parameter
of the system.

There exists natural extensions of this work. One direction is to
generalize the synchronization analysis to networks of
non-identical units. Here the problem is more challenging, since
complete synchronization is usually not possible, and one must
look for other suitable types of solutions and investigate their
stability. Other directions for extension include
higher-dimensional and/or continuous-time systems. In all cases,
the possibility of the emergence of new dynamics via
synchronization offers a helpful perspective in our efforts to
understand complex systems.

\begin{appendix}\section{Definition of a full family of S-unimodal maps}\label{Appendix}
\begin{definition}[Unimodal Map] Let $f:I=[0,1]\rightarrow I=[0,1]$. The map
is unimodal if \begin{enumerate}\item $f(0) = f(1) = 0$ \item $f$
has a unique critical point $c$ (i.e. $f'(c) =0$) with $0< c< 1$.
\end{enumerate}
\end{definition}

\begin{definition}[Schwarzian Derivative]
The Schwarzian derivative of a function $f$ at $x$ is \[Sf(x) =
\frac{f'''(x)}{f'(x)} - \frac{3}{2}
\left(\frac{f''(x)}{f'(x)}\right)^2.\]
\end{definition}
The important role of a negative Schwarzian derivative was
discovered by Singer \cite{Singer78}. A negative Schwarzian
derivative restricts the number of stable periodic orbits. Singer
showed that a unimodal map whose Schwarzian derivative is negative
has at most one stable periodic orbit. Furthermore, if the critical
point is not attracted to a stable periodic orbit then the map has
no stable orbit at all.
\begin{definition}[Itinerary]
Let $x \in I$. The itinerary of $x$ under $f$ is the infinite
sequence $\mathcal{S}(x) = (s_0s_1s_2,...)$ where \[ s_j = \left\{
\begin{array}{r@{\quad if \quad}l} 0 & f^j(x) < c \\ 1 & f^j(x) >
c \\ C & f^j(x) = c.
\end{array} \right.
\]
\end{definition}
The idea of kneading sequences goes back to Milnor and Thurston
\cite{Milnor77}.\begin{definition} [Kneading Sequence] The kneading
sequence $K(f)$ of $f(x)$ is the itinerary of $f(c)$, i.e., $K(f) =
\mathcal{S}(f(c))$.
\end{definition}
\begin{definition}[Full family of S-unimodal maps] \label{FullFamily} Let $f_\mu$ be a family of unimodal maps with
$\mu_0 \leq \mu \leq \mu_1$. $f_\mu$ is called a full family of
S-unimodal maps if
\begin{enumerate}\item $f_{\mu_0}(x) \equiv 0$ for all $x \in I$.
\item When $\mu = \mu_1$, $K(f_\mu) = (100\overline{0}...)$. \item
$Sf_\mu(x) < 0 $ for all $\mu > \mu_0$ and $x\in I$.
\end{enumerate}
\end{definition}

\section{Proof of Theorem \ref{38}}\label{Appendix B}
\begin{proof} It is straightforward to show
that $f_\mu$ is unimodal for all $\mu$. The first two points in
Definition \ref{FullFamily} obviously hold because $f^2_{\mu_1}(c) =
0$ and $0$ is a fixed point. So we only have to prove that
 $Sf_\mu(x) < 0 $ for all $\mu > \mu_0= 0$.
We calculate the following derivatives.
\begin{eqnarray*}&f'_\mu(x)& = \frac{1
-e^{-\mu}}{1 + e^{-\mu}}\Theta - \frac{2\mu \Theta e^{-\mu
x}}{(1+e^{-\mu x})^2}\\
&f''_\mu(x)& = 2\mu^2 \Theta e^{-\mu x} \left[ \frac{1}{(1+e^{-\mu
x})^2} -
\frac{2e^{-\mu x}}{(1+e^{-\mu x})^3}\right]\\
&f'''_\mu(x)&= 2\mu^3 \Theta e^{-\mu x} \left[ -\frac{1}{(1+e^{-\mu
x})^2} + \frac{6e^{-\mu x}}{(1+e^{-\mu x})^3} -\frac{6e^{-2\mu
x}}{(1 +e^{-\mu x})^4}\right].
\end{eqnarray*}  Putting $a := \frac{1
-e^{-\mu}}{1 + e^{-\mu}}$ and $b := a(1+e^{-\mu x})^2 - 2\mu e^{-\mu
x}$, we have:
\begin{eqnarray*}&\frac{f'''_\mu(x)}{f'_\mu(x)}& = \begin{array}{l} \frac{1}{b^2}\left[-2a\mu^3 e^{-\mu x} (1 +
e^{-\mu x})^2 + 4 \mu ^4 e^{-2\mu x}+ 12 a\mu^3 e^{-2\mu x}(1 +
e^{-\mu x}) \right.
\\ - \left.\frac{24 \mu^4 e^{-3\mu x}}{1 + e^{-\mu x}}
-12 a\mu^3 e^{-3\mu x} + \frac{24\mu^4 e^{-4\mu x}}{(1 + e^{-\mu
x})^2}\right]\end{array}\\&-\frac{3}{2}\left(\frac{f''_\mu(x)}{f'_\mu(x)}\right)^2&
= \frac{1}{b^2} \left[-6\mu^4e^{-2\mu x} + \frac{24\mu^4e^{-3\mu
x}}{1 + e^{-\mu x}} - \frac{24\mu^4e^{-4\mu x}}{(1 + e^{-\mu x})^2}
\right]
\end{eqnarray*}Thus, the Schwarzian derivative
\begin{eqnarray*}Sf_\mu(x) &=& \frac{f'''_\mu(x)}{f'_\mu(x)}
-\frac{3}{2}\left(\frac{f''_\mu(x)}{f'_\mu(x)}\right)^2\\ &=&
\underbrace{\frac{-2a\mu^3e^{-\mu x}}{b^2}}_{<0} \underbrace{\left[1
+ e^{-2\mu x} + \left(\frac{\mu}{a} -4\right)e^{-\mu x} \right]}_{
=: g_\mu(x)}
\end{eqnarray*}
is negative if $g_\mu(x)
> 0$. The critical point of $g_\mu(x)$ is given by
\[c = \frac{\ln\left( 2- \frac{\mu}{2a} \right)}{-\mu}.\]
It is easy to verify that this critical point $c$ is actually a
minimum of the function $g_\mu(x)$. First, we assume that $c\in
(0,1)$. This leads to
\begin{equation} \label{35} e^{-\mu} <  2- \frac{\mu}{2a} <1 .\end{equation}
Inserting the critical point yields:
\begin{eqnarray*}g_\mu(c) &=& 1+ \left( 2-\frac{\mu}{2a}\right)^2+\left(\frac{\mu}{a} -4 \right)\left(2-
\frac{\mu}{2a}\right) \\&=& -3 + 2\frac{\mu}{a} -
\left(\frac{\mu}{2a}\right)^2\\ &=& 1 - \underbrace{\left(2 -
\frac{\mu}{2a}\right)^2}_{< \,1}> 0
\end{eqnarray*}
So we only have to check that $g_\mu(x)>0$ at the boundary. First we
consider the case $x = 0$.
\begin{eqnarray*} g_\mu(0) = -2 + \frac{\mu}{a} = -2 + \frac{1
+e^{-\mu}}{1 - e^{-\mu}}\mu
\end{eqnarray*}
Calculating the minimum of $g_\mu(0)$ with respect to $\mu$ yields:
\begin{eqnarray*}\frac{dg_\mu(0)}{d\mu} &=& \frac{1 + (1-\mu)e^{-\mu}}{1 -
e^{-\mu}} - \frac{\mu e^{-\mu}(1+e^{-\mu})}{(1- e^{-\mu})^2} \stackrel{!}{=} 0 \\
&\Rightarrow& 1  -2\mu e^{-\mu} - e^{-2\mu} = 0
\end{eqnarray*}
Hence, we see that $\mu = 0$ is the minimum. Using l'Hospital's rule
we obtain:
\[ \lim_{\mu\searrow 0}g_\mu(0) =  -2 + \lim_{\mu \searrow 0}  \frac{\mu (1+ e^{-\mu})}{1-e^{-\mu}}
=-2 +  \lim_{\mu \searrow 0} \frac{1 + e^{-\mu} - \mu
e^{-\mu}}{e^{-\mu}} \searrow 0\] This shows that  $g(0) > 0$ if $\mu
>0$.
The case $x = 1$ is treated in the same way.  Thus the family
(\ref{54}) is a full family of S-unimodal maps.
\end{proof}
\end{appendix}

\end{document}